\documentclass[preprint,12pt]{elsarticle}
\usepackage[margin=2.5cm]{geometry}
\usepackage{setspace}
\usepackage{amsmath,amssymb,amsthm,graphicx,booktabs,tabularx,array}
\usepackage{tikz,helvet}
\usetikzlibrary{arrows.meta,calc}
\usepackage{caption}
\usepackage{placeins}
\usepackage[hidelinks]{hyperref}
\hypersetup{pdfauthor={Chang Liu, Hao Shen and Chuanying Li}}
\biboptions{numbers,sort&compress}
\providecommand{\myfooterfont}[1]{}
\myfooterfont{\scriptsize\itshape}
\newtheorem{proposition}{Proposition}
\newcommand{\R}{\mathbb R}
\newcommand{\E}{\mathbb E}
\newcommand{\Yamp}{Y_{\mathrm{amp}}}
\newcommand{\result}[1]{\csname result@#1\endcsname}
\expandafter\def\csname result@assoc_betti_proxy\endcsname{0.925}
\expandafter\def\csname result@assoc_off_burn_weighted_pstag_gbar\endcsname{0.928}
\expandafter\def\csname result@baseline_radius_ood_gp_et\endcsname{4.91}
\expandafter\def\csname result@baseline_radius_ood_moe_et\endcsname{4.90}
\expandafter\def\csname result@baseline_random_gp_et\endcsname{7.37}
\expandafter\def\csname result@baseline_random_moe_et\endcsname{7.84}
\expandafter\def\csname result@be_radius_ood_1000_burn_gp_burn_moe_gain_pct\endcsname{-12.37}
\expandafter\def\csname result@be_radius_ood_1000_design_gp_burn_gp_gain_pct\endcsname{0.64}
\expandafter\def\csname result@be_radius_ood_250_design_gp_burn_gp_gain_pct\endcsname{5.19}
\expandafter\def\csname result@be_random_1000_burn_gp_burn_moe_gain_pct\endcsname{-40.21}
\expandafter\def\csname result@be_random_1000_design_gp_burn_gp_gain_pct\endcsname{43.72}
\expandafter\def\csname result@be_random_250_design_gp_burn_gp_gain_pct\endcsname{26.98}
\expandafter\def\csname result@burn_random_rmse\endcsname{0.285}
\expandafter\def\csname result@ch_paired_wall_median\endcsname{7.43}
\expandafter\def\csname result@dense_expert_gain_random_burn3_125\endcsname{15.53}
\expandafter\def\csname result@dense_representation_gain_radius_ood_gp_125\endcsname{18.97}
\expandafter\def\csname result@dense_representation_gain_radius_ood_moe_125\endcsname{20.84}
\expandafter\def\csname result@dense_representation_gain_random_gp_125\endcsname{20.45}
\expandafter\def\csname result@dense_representation_gain_random_moe_125\endcsname{24.26}
\expandafter\def\csname result@e1_radius_ood_gp_betti3_mse\endcsname{0.0853}
\expandafter\def\csname result@e1_radius_ood_gp_betti_vs_burn\endcsname{8.63}
\expandafter\def\csname result@e1_radius_ood_gp_burn3_gain\endcsname{52.98}
\expandafter\def\csname result@e1_radius_ood_gp_burn3_mse\endcsname{0.0934}
\expandafter\def\csname result@e1_radius_ood_gp_current3_mse\endcsname{0.1262}
\expandafter\def\csname result@e1_radius_ood_gp_design_mse\endcsname{0.1986}
\expandafter\def\csname result@e1_radius_ood_gp_oracle_burn_mse\endcsname{0.0106}
\expandafter\def\csname result@e1_radius_ood_gp_rand3_mse\endcsname{0.1988}
\expandafter\def\csname result@e1_radius_ood_gp_shuffle3_mse\endcsname{0.1987}
\expandafter\def\csname result@e1_radius_ood_moe_betti3_mse\endcsname{0.0861}
\expandafter\def\csname result@e1_radius_ood_moe_betti_vs_burn\endcsname{7.83}
\expandafter\def\csname result@e1_radius_ood_moe_burn3_gain\endcsname{41.84}
\expandafter\def\csname result@e1_radius_ood_moe_burn3_mse\endcsname{0.0934}
\expandafter\def\csname result@e1_radius_ood_moe_current3_mse\endcsname{0.1144}
\expandafter\def\csname result@e1_radius_ood_moe_design_mse\endcsname{0.1606}
\expandafter\def\csname result@e1_radius_ood_moe_oracle_burn_mse\endcsname{0.0097}
\expandafter\def\csname result@e1_radius_ood_moe_rand3_mse\endcsname{0.1607}
\expandafter\def\csname result@e1_radius_ood_moe_shuffle3_mse\endcsname{0.1624}
\expandafter\def\csname result@e1_random_gp_betti3_mse\endcsname{0.0795}
\expandafter\def\csname result@e1_random_gp_betti_vs_burn\endcsname{1.86}
\expandafter\def\csname result@e1_random_gp_burn3_gain\endcsname{60.25}
\expandafter\def\csname result@e1_random_gp_burn3_mse\endcsname{0.0810}
\expandafter\def\csname result@e1_random_gp_current3_mse\endcsname{0.1200}
\expandafter\def\csname result@e1_random_gp_design_mse\endcsname{0.2037}
\expandafter\def\csname result@e1_random_gp_oracle_burn_mse\endcsname{0.0109}
\expandafter\def\csname result@e1_random_gp_rand3_mse\endcsname{0.2040}
\expandafter\def\csname result@e1_random_gp_shuffle3_mse\endcsname{0.2038}
\expandafter\def\csname result@e1_random_moe_betti3_mse\endcsname{0.0798}
\expandafter\def\csname result@e1_random_moe_betti_vs_burn\endcsname{0.93}
\expandafter\def\csname result@e1_random_moe_burn3_gain\endcsname{53.84}
\expandafter\def\csname result@e1_random_moe_burn3_mse\endcsname{0.0806}
\expandafter\def\csname result@e1_random_moe_current3_mse\endcsname{0.1101}
\expandafter\def\csname result@e1_random_moe_design_mse\endcsname{0.1745}
\expandafter\def\csname result@e1_random_moe_oracle_burn_mse\endcsname{0.0113}
\expandafter\def\csname result@e1_random_moe_rand3_mse\endcsname{0.1765}
\expandafter\def\csname result@e1_random_moe_shuffle3_mse\endcsname{0.1746}
\expandafter\def\csname result@e6_radius_1000_expert\endcsname{39.14}
\expandafter\def\csname result@e6_radius_1000_gp_gain\endcsname{14.20}
\expandafter\def\csname result@e6_radius_1000_moe_gain\endcsname{20.76}
\expandafter\def\csname result@e6_radius_320_expert\endcsname{-4.05}
\expandafter\def\csname result@e6_radius_320_gp_gain\endcsname{42.36}
\expandafter\def\csname result@e6_radius_320_moe_gain\endcsname{40.62}
\expandafter\def\csname result@e6_random_1000_expert\endcsname{10.61}
\expandafter\def\csname result@e6_random_1000_gp_gain\endcsname{41.34}
\expandafter\def\csname result@e6_random_1000_moe_gain\endcsname{35.35}
\expandafter\def\csname result@e6_random_320_expert\endcsname{-14.16}
\expandafter\def\csname result@e6_random_320_gp_gain\endcsname{22.62}
\expandafter\def\csname result@e6_random_320_moe_gain\endcsname{23.76}
\expandafter\def\csname result@expert_random_burn3_1000\endcsname{0.52}
\expandafter\def\csname result@expert_random_design_1000\endcsname{14.33}
\expandafter\def\csname result@expert_random_design_250\endcsname{6.87}
\expandafter\def\csname result@grid_200_log10_yamp_max_abs_difference\endcsname{0.3914}
\expandafter\def\csname result@grid_200_log10_yamp_median_abs_difference\endcsname{0.0570}
\expandafter\def\csname result@grid_200_off_burn_weighted_pstag_gbar_max_abs_relative_pct\endcsname{90.3}
\expandafter\def\csname result@grid_200_off_burn_weighted_rhordt_max_abs_relative_pct\endcsname{12.2}
\expandafter\def\csname result@grid_200_off_burn_weighted_tion_kev_max_abs_relative_pct\endcsname{80.7}
\expandafter\def\csname result@grid_200_off_neutrons_max_abs_relative_pct\endcsname{35.9}
\expandafter\def\csname result@grid_400_log10_yamp_max_abs_difference\endcsname{0.3597}
\expandafter\def\csname result@grid_400_log10_yamp_median_abs_difference\endcsname{0.0230}
\expandafter\def\csname result@grid_400_off_burn_weighted_pstag_gbar_max_abs_relative_pct\endcsname{89.0}
\expandafter\def\csname result@grid_400_off_burn_weighted_rhordt_max_abs_relative_pct\endcsname{4.4}
\expandafter\def\csname result@grid_400_off_burn_weighted_tion_kev_max_abs_relative_pct\endcsname{80.4}
\expandafter\def\csname result@grid_400_off_neutrons_max_abs_relative_pct\endcsname{12.4}
\expandafter\def\csname result@grid_800_log10_yamp_max_abs_difference\endcsname{0.1599}
\expandafter\def\csname result@grid_800_log10_yamp_median_abs_difference\endcsname{0.0056}
\expandafter\def\csname result@grid_800_off_burn_weighted_pstag_gbar_max_abs_relative_pct\endcsname{87.3}
\expandafter\def\csname result@grid_800_off_burn_weighted_rhordt_max_abs_relative_pct\endcsname{1.4}
\expandafter\def\csname result@grid_800_off_burn_weighted_tion_kev_max_abs_relative_pct\endcsname{76.1}
\expandafter\def\csname result@grid_800_off_neutrons_max_abs_relative_pct\endcsname{3.6}
\expandafter\def\csname result@illustrative_break_even_paired\endcsname{63}
\expandafter\def\csname result@latency_1\endcsname{0.0139}
\expandafter\def\csname result@latency_987\endcsname{0.1438}
\expandafter\def\csname result@learn_radius_ood_gp_1000\endcsname{52.8}
\expandafter\def\csname result@learn_radius_ood_moe_1000\endcsname{41.1}
\expandafter\def\csname result@learn_random_gp_1000\endcsname{60.2}
\expandafter\def\csname result@learn_random_moe_1000\endcsname{53.8}
\expandafter\def\csname result@package_fit_seconds\endcsname{928.7}
\expandafter\def\csname result@primary_expert_radius_ood\endcsname{-0.01}
\expandafter\def\csname result@primary_expert_random\endcsname{0.52}
\expandafter\def\csname result@state_radius_ood_pressure\endcsname{0.202}
\expandafter\def\csname result@state_radius_ood_rhor\endcsname{0.279}
\expandafter\def\csname result@state_radius_ood_temperature\endcsname{0.211}
\expandafter\def\csname result@state_random_pressure\endcsname{0.165}
\expandafter\def\csname result@state_random_rhor\endcsname{0.235}
\expandafter\def\csname result@state_random_temperature\endcsname{0.168}
\expandafter\def\csname result@tail_radius_ood_gp\endcsname{1.0}
\expandafter\def\csname result@tail_radius_ood_moe\endcsname{4.2}
\expandafter\def\csname result@tail_random_gp\endcsname{35.5}
\expandafter\def\csname result@tail_random_moe\endcsname{28.2}
\expandafter\def\csname result@thermal_guard_count\endcsname{713}
\expandafter\def\csname result@thermal_max_guard_continuation_difference\endcsname{1.19\times10^{-11}}
\expandafter\def\csname result@thermal_max_target_tolerance_difference\endcsname{4.75\times10^{-5}}
\expandafter\def\csname result@thermal_n_audit\endcsname{128}
\expandafter\def\csname result@thermal_n_thermal\endcsname{3000}
\expandafter\def\csname result@thermal_solver_ms\endcsname{2.031}

\journal{Computer Physics Communications}
\date{October 10, 2026}
\begin{document}
\begin{frontmatter}
\title{Physical-state supervision for design-only Gaussian process surrogates of implosion simulations}
\author[iapcm]{Chang Liu\corref{cor1}}
\ead{liuchang@iapcm.ac.cn}
\cortext[cor1]{Corresponding author.}
\author[iapcm]{Hao Shen}
\author[iapcm]{Chuanying Li}
\address[iapcm]{Institute of Applied Physics and Computational Mathematics, Beijing, P.R. China}
\begin{abstract}
Physical diagnostics recorded during simulation can organize a sharply varying response, but using them as surrogate inputs normally requires another calculation at prediction. We construct a physical-state-supervised Gaussian process surrogate that learns these diagnostics from design variables. Cross-fitted state estimates augment a retained design representation; a global Gaussian process and an optional soft mixture of response-band experts use identical inputs. On a database of 1,987 paired one-dimensional implosions, learned burn-weighted states from calculations without alpha-particle heating reduce mean squared log-response error by \result{e1_random_gp_burn3_gain}\% and \result{e1_radius_ood_gp_burn3_gain}\% for the global predictor on random holdout and design-space extrapolation. Random-projection and shuffled-state supervision do not reproduce these gains. For CH, matched-input comparisons show that experts help at smaller training budgets but contribute little at 1,000 designs with this representation. A thermal-feedback benchmark validates the complete learned-state procedure, while 3,000 newly calculated Be/Cocktail designs expose weaker extrapolation gains and cases where experts degrade accuracy. A conditional state-error bound separates encoding sensitivity from routing effects. An open Python implementation reproduces the archived predictions and is released with the data, per-design results, and measured costs. Physical-state supervision is thus a useful representation for finite-data surrogates, whereas expert specialization should be selected by validation.
\end{abstract}
\end{frontmatter}

\section{Introduction}\label{sec:introduction}
An implosion calculation produces more than a fusion yield. It also records the assembly of the fuel, its thermodynamic state, and the evolution of confinement. These quantities enter physical descriptions of inertial confinement fusion (ICF), where compression, heating, and energy deposition determine whether a change in drive produces a modest or amplified burn response~\cite{lindl1995,betti2016}. Experimental progress from fuel gain to burning plasma and target gain emphasizes the importance of this feedback~\cite{hurricane2014,zylstra2022,abushawareb2024}. A surrogate trained only on design--yield pairs must infer their combined influence indirectly. The difficulty is pronounced in a finite database spanning a rapid change in yield amplification: nearby designs can have markedly different responses, while a global fit must also describe broad regions with weaker variation.

Using recorded states as additional regression inputs can reduce this difficulty, but it changes the prediction task if those states require a simulation for each new design. The distinction matters in design searches, where avoiding that calculation is the purpose of the surrogate. Here the physical states are available as auxiliary training labels. A separate model predicts them from a new design before the response model is evaluated. The question is whether the resulting approximation retains enough of their predictive value to improve a surrogate whose deployment input remains the design alone.

Physical knowledge can enter a learning problem through several different constructions~\cite{karniadakis2021,karpatne2017}. Physics-informed neural networks incorporate differential-equation residuals and initial or boundary conditions into training~\cite{raissi2019}. Multifidelity methods combine outputs from models of different cost and accuracy, through autoregressive or nonlinear relationships and other information-fusion strategies~\cite{kennedy2000,perdikaris2017,peherstorfer2018}. In ICF, data-driven modeling has also addressed multimodal outputs and discrepancies between simulation and experiment~\cite{hatfield2021,humbird2020}. Anirudh et al.~\cite{anirudh2020} use manifold and cycle consistencies to constrain an ICF surrogate and its inverse. These approaches establish that physical structure and auxiliary representations can guide learning; their particular constraints and prediction interfaces differ from the scalar-response problem considered here.

Training with information that is absent at prediction is an established learning setting~\cite{vapnik2009,vapnik2015,lopezpaz2016}. In the present construction, the additional information consists of named physical diagnostics, and the training procedure explicitly learns a map to these quantities. Auxiliary tasks can shape a shared representation in multitask learning~\cite{caruana1997}; here the physical outputs are learned explicitly in a separate stage. The original design descriptors remain available to the response model, so the diagnostics need not form a sufficient, closed description of the response.

Gaussian processes (GPs) provide a useful regression model at the sample sizes of interest~\cite{rasmussen2006}. Their behavior depends on the representation in which covariance is defined. Nonstationary responses can be treated by partitioning the input domain, as in treed GPs~\cite{gramacy2008}, by learning input transformations~\cite{snoek2014}, or by constructing deep kernels~\cite{wilson2016}. Learning a few physically supervised coordinates offers another way to change the geometry seen by a conventional kernel while retaining an explicit interpretation of the auxiliary outputs.

Local regression is complementary to this representation choice. Mixtures of experts assign different predictors to different parts of a response map~\cite{jacobs1991,jordan1994}, and mixtures of GP experts have a substantial history~\cite{rasmussen2002}. Moustapha and Sudret~\cite{moustapha2023} combine clustering, classification, and local GP regression for nonstationary and discontinuous functions, a direct precedent for classifier-based regional GP regression. Physical information can also act inside a mixture: Pitchforth et al.~\cite{pitchforth2026} embed physical relations in kernels and use switching variables to vary their contributions. Our construction instead uses standard GP kernels on a learned physical representation and combines their predictive means through a gate.

The two modeling choices---learning a representation and introducing experts---must be evaluated separately. A mixture with additional state information can outperform a design-based global model even when partitioning contributes little. Conversely, an informative representation need not eliminate local differences that matter at smaller sample sizes. Comparisons that change both inputs and architecture cannot establish which effect is responsible. We therefore use the same representation in the global GP, local experts, and gate, and compare global and mixture predictors at each representation and training budget.

This paper develops a physical-state-supervised GP architecture and evaluates expert specialization as an optional extension. The reference states come from calculations with alpha-particle energy deposition disabled; they are learned from designs and are never queried from a solver during inference. We give a conditional bound for errors transmitted through the predicted states and a separate decomposition of regression and routing errors. Numerical comparisons then assess the net value of learned states, the choice of physical supervision, and the additional contribution of experts under matched inputs. Random-projection and shuffled-state controls test whether the improvement can be reproduced by the same auxiliary pathway without the physical correspondence. A finite-fuel thermal model tests the complete learned-state procedure, and a newly reconstructed Be/Cocktail database examines its limits in another target family. The implementation and numerical record are openly released.

\section{Problem setting}\label{sec:problem}
Let $x\in\mathcal X\subset\R^d$ denote a design and let $y=f(x)$ be the scalar response of interest. The training database is
\begin{equation}
 \mathcal D_N=\{(x_i,z_i,y_i)\}_{i=1}^N,
 \qquad z_i=g_{\mathrm{ref}}(x_i)\in\R^q,
 \label{eq:database}
\end{equation}
where $g_{\mathrm{ref}}$ extracts transformed physical diagnostics from a reference calculation. A new query provides $x$, but neither $g_{\mathrm{ref}}(x)$ nor $f(x)$. We first construct deterministic design descriptors $b(x)=[x,\phi(x)]$ and learn a state map $E_\theta$. The response predictor receives
\begin{equation}
 \widehat z(x)=E_\theta(b(x)),\qquad
 u(x)=[b(x),\widehat z(x)].
 \label{eq:representation}
\end{equation}
The auxiliary labels in \eqref{eq:database} change only the training information: once training is fixed, $\widehat z$ is a deterministic function of $x$, and any benefit comes from a more useful representation for finite-sample regression.

For the ICF application, the response is the logarithm of the neutron-yield amplification,
\begin{equation}
 y=\log_{10}\Yamp,
 \qquad \Yamp=\frac{Y_{\mathrm{on}}}{Y_{\mathrm{off}}}.
 \label{eq:response}
\end{equation}
The on calculation retains alpha-particle heating; the off calculation disables this energy deposition while retaining the other modeled assembly and transport processes. This ratio is the yield amplification used to quantify alpha heating in ICF~\cite{betti2015}. Writing $\log_{10}Y_{\mathrm{on}}=\log_{10}Y_{\mathrm{off}}+y$ separates the no-feedback yield, which reflects compression and assembly, from the feedback response that produces the rapid change studied here. Because the off calculation is a numerical counterfactual, the construction applies to simulation databases rather than to measured yields.

We compare two three-component state sets. The \emph{Peak} set contains implosion-speed magnitude, in-flight aspect ratio, and fuel adiabat from the off calculation at peak implosion speed. The \emph{Burn} set contains the recorded off burn-weighted areal density, ion temperature, and stagnation-pressure diagnostic. Each recorded quantity is transformed by $\log_{10}|s_j|$ in the fixed units used by its database column; zero or nonfinite entries are treated as missing. The Burn variables relate to confinement, thermal conditions, and compression, whereas the Peak variables describe an earlier assembly diagnostic.

The CH representations in Table~\ref{tab:information} share a 20-component design vector. The map $\phi$ adds 25 algebraic descriptors: radiation-drive proxies, relative pulse amplitudes and durations, approximate mass loading, and geometric and density ratios. These require no solver evaluation, so the \emph{Design} comparator already includes physical information. Predicted states enlarge the response input from 45 to 48 components.

\begin{table}[t]
\centering\small
\caption{CH representations and information access. The CH prediction adapter takes the same 20 raw design variables in every row. The global GP and GP--MoE use identical representations within a row. State-supervised rows use additional offline training labels.}\label{tab:information}
\begin{tabularx}{\linewidth}{@{}l c X X@{}}
\toprule
Representation & Dimension & Auxiliary training labels & State input at prediction\\
\midrule
Design & 45 & None & None\\
Peak states & 48 & Off-state speed, IFAR and adiabat at peak speed & Encoder predictions\\
Burn states & 48 & Off-state burn-weighted $\rho R$, $T_i$ and $P_{\mathrm{stag}}$ & Encoder predictions\\
\bottomrule
\end{tabularx}
\end{table}

\section{Physical-state-supervised GP architecture}\label{sec:method}
The core predictor is a state encoder followed by a global GP on the augmented representation in \eqref{eq:representation}. Both stages are trained sequentially. Figure~\ref{fig:architecture} separates this prediction path from an optional expert extension, which uses the same inputs and adds a learned gate.

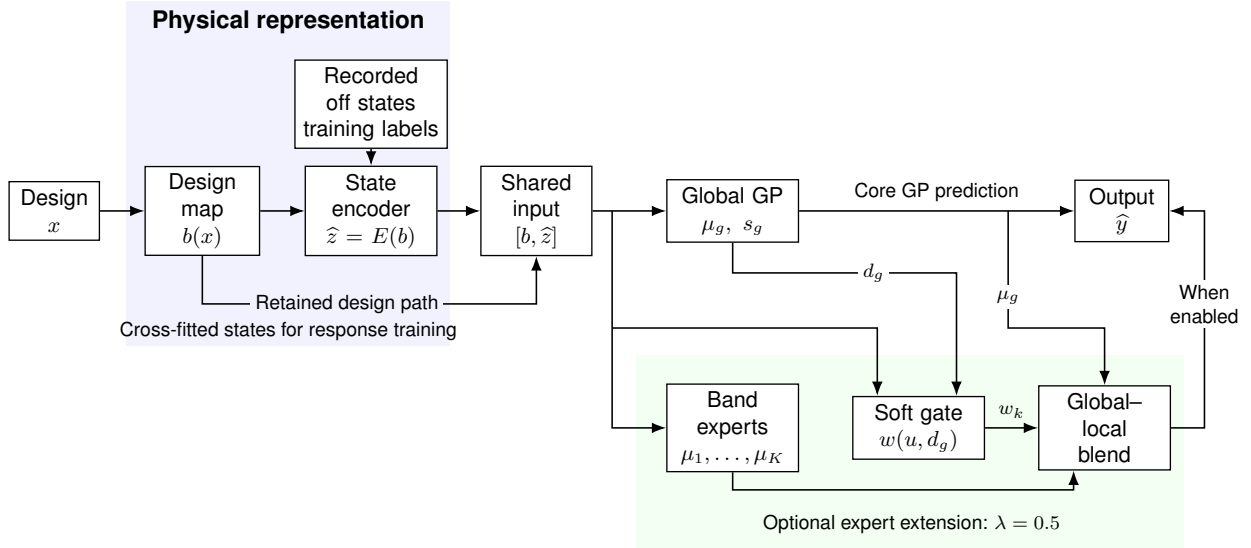
\begin{figure}[t]
\centering
\resizebox{\linewidth}{!}{\begingroup\setstretch{1}
\begin{tikzpicture}[x=1cm,y=1cm,>=Latex,
 box/.style={draw=black,line width=.55pt,fill=white,align=center,minimum height=.85cm,text width=1.7cm,inner sep=3pt,font=\sffamily\fontsize{9}{11}\selectfont},
 flow/.style={->,draw=black,line width=.65pt},
 training/.style={->,draw=black,densely dashed,line width=.65pt},
 note/.style={font=\sffamily\fontsize{8}{10}\selectfont,align=center}]
\fill[blue!5] (1.45,1.15) rectangle (6.15,6.15);
\fill[green!5] (8.85,-1.8) rectangle (16.8,1.0);
\node[font=\sffamily\bfseries\fontsize{10}{12}\selectfont] at (3.8,5.85) {Physical representation};
\node[box,text width=1.1cm] (x) at (.4,3.1) {Design\\$x$};
\node[box,text width=1.45cm] (b) at (2.55,3.1) {Design map\\$b(x)$};
\node[box,text width=1.7cm] (enc) at (5,3.1) {State encoder\\$\widehat z=E(b)$};
\node[box,text width=2cm] (z) at (5,4.65) {Recorded off states\\training labels};
\node[box,text width=1.4cm] (u) at (7.4,3.1) {Shared input\\$[b,\widehat z]$};
\node[box] (gp) at (10.25,3.1) {Global GP\\$\mu_g,\ s_g$};
\node[box,text width=1.15cm] (y) at (15.9,3.1) {Output\\$\widehat y$};
\node[box] (experts) at (10.25,-.05) {Band experts\\$\mu_1,\ldots,\mu_K$};
\node[box] (gate) at (12.95,-.05) {Soft gate\\$w(u,d_g)$};
\node[box] (blend) at (15.65,-.05) {Global--local\\blend};
\draw[flow] (x)--(b);\draw[flow] (b)--(enc);\draw[flow] (enc)--(u);
\draw[training] (z.south)--(enc.north);
\draw[flow] (b.south)--(2.55,1.75)--(7.4,1.75)--(u.south);
\node[note,fill=blue!5,inner sep=2pt] at (4.65,1.75) {Retained design path};
\draw[flow] (u.east)--(gp.west);
\draw[flow] (gp.east)--node[note,above]{Core GP prediction}(y.west);
\coordinate (branch) at (8.5,3.1);
\draw[flow] (branch)--(8.5,-.05)--(experts.west);
\draw[flow] (8.5,1.4)--(12.35,1.4)--($(gate.north)+(-.6,0)$);
\draw[flow] (gp.south)--(10.25,2.2)--(13.5,2.2)--($(gate.north)+(.55,0)$);
\node[note,fill=white,inner sep=2pt] at (12.3,2.2) {$d_g$};
\draw[flow] (14.25,3.1)--(14.25,1.4)--(15.65,1.4)--(blend.north);
\node[note,fill=white,inner sep=2pt] at (14.25,1.85) {$\mu_g$};
\draw[flow] (gate.east)--node[note,above]{$w_k$}(blend.west);
\draw[flow] (experts.south)--(10.25,-.95)--(15.2,-.95)--($(blend.south)+(-.45,0)$);
\draw[flow] (blend.east)--(17.1,-.05)--(17.1,3.1)--(y.east);
\node[note,fill=white,inner sep=2pt] at (17.1,1.75) {When\\enabled};
\node[note] at (12.85,-1.45) {Optional expert extension: $\lambda=0.5$};
\node[note] at (3.8,1.35) {Cross-fitted states for response training};
\end{tikzpicture}
\endgroup}
\caption{Physical-state supervision with a global GP and an optional expert branch. The upper path predicts the response directly from design variables through a learned physical representation. Recorded off states supervise the encoder only during training; the retained design path and predicted states form the common response input. The shaded lower branch adds response-band experts, global predictive diagnostics, and soft blending when enabled.}\label{fig:architecture}
\end{figure}

\subsection{Cross-fitted state supervision}\label{sec:encoder}
A direct in-sample fit of the state encoder can supply unrealistically accurate auxiliary inputs to the response model. We therefore partition the outer training set into five folds $I_1,\ldots,I_5$ and construct
\begin{equation}
 \widehat z_i^{\mathrm{cf}}=E_{\theta^{(-j)}}(b(x_i)),
 \qquad i\in I_j,
 \label{eq:crossfit}
\end{equation}
where $E_{\theta^{(-j)}}$ is fitted without fold $I_j$. Missing state labels are filled with medians computed inside the encoder's fitting subset. The encoder is a multi-output GP with a common input covariance and separately normalized output columns; this is a simpler covariance structure than multi-task GP models with learned output correlations~\cite{bonilla2008,alvarez2012}. After the cross-fitted predictions have been assembled, a final encoder is fitted on the complete outer training set for deployment.

The response models are trained on $u_i^{\mathrm{cf}}=[b(x_i),\widehat z_i^{\mathrm{cf}}]$. Test representations use the full-training encoder. The use of predictions from held-out training folds is related to stacked regression~\cite{wolpert1992,breiman1996}. Cross-fitting makes the training inputs, like the test inputs, estimated rather than recorded states, although each fold encoder sees fewer designs than the deployment encoder. Both stages use only outer-training observations.

\subsection{Global response regression}\label{sec:regression}
The encoder and CH response GPs use a Mat\'ern-$3/2$ covariance with automatic relevance determination. On standardized inputs $u,v$, the covariance is
\begin{align}
 k(u,v)&=\sigma_f^2(1+\sqrt3\,d_\ell(u,v))
              \exp[-\sqrt3\,d_\ell(u,v)]+\sigma_n^2\delta_{uv},\label{eq:kernel}\\
 d_\ell(u,v)^2&=\sum_{j=1}^{p}\frac{(u_j-v_j)^2}{\ell_j^2}.
\end{align}
Feature imputation and standardization are fitted on the relevant training subset. The amplitude and length scales are initialized at one and bounded between $10^{-2}$ and $10^2$; the white-noise variance is initialized at $10^{-3}$ and bounded between $10^{-6}$ and one. Marginal-likelihood optimization uses no random restarts, and $10^{-6}$ is added to the covariance diagonal for numerical regularization. Outputs are normalized during fitting.

The response GP is fitted to $(u_i^{\mathrm{cf}},y_i)$. At prediction, it returns $\widehat f_{\mathrm{GP}}(x)=\mu_g(u(x))$, with posterior scale $s_g$ available for diagnostics.

\subsection{Training and design-only prediction}\label{sec:algorithm}
The core model stores the deployment state encoder and the global response GP. The CH adapter constructs deterministic descriptors from the raw design; the generic interface accepts a fixed design schema. Algorithm~1 requires physical labels only for training.

\begin{center}\begin{minipage}{.96\linewidth}\small
\hrule\vspace{5pt}
\textbf{Algorithm 1. Physical-state-supervised global GP}\par
\emph{Training input:} designs $x_i$, transformed reference states $z_i$, responses $y_i$, and a fixed split seed.
\begin{enumerate}\setlength{\itemsep}{1pt}\setlength{\parskip}{0pt}
\item Construct deterministic design descriptors $b(x_i)$.
\item Fit five fold encoders with fold-local state imputation; assemble $\widehat z_i^{\mathrm{cf}}$ by \eqref{eq:crossfit}.
\item Fit the deployment encoder on all training designs and state labels.
\item Fit the global response GP to $u_i^{\mathrm{cf}}=[b(x_i),\widehat z_i^{\mathrm{cf}}]$ and $y_i$.
\item For a new design $x$, compute $u(x)=[b(x),E_\theta(b(x))]$ and return $\mu_g(u(x))$.
\end{enumerate}
\emph{Optional extension:} fit response-band experts and a predictive gate as specified in Section~\ref{sec:gate}; use \eqref{eq:predictor} when this extension is enabled.
\vspace{3pt}\hrule
\end{minipage}\end{center}

\subsection{Optional response-band experts}\label{sec:gate}
Experts are assigned operational response bands. With $c_1=\log_{10}2$ and $c_2=\log_{10}20$, the training label is
\begin{equation}
 r(y)=\begin{cases}
 1,&y\le c_1,\\
 2,&c_1<y\le c_2,\\
 3,&y>c_2.
 \end{cases}
 \label{eq:bands}
\end{equation}
These labels train the experts and gate; at inference their probabilities are estimated from the design.

Expert $k$ is fitted only to designs with training label $r_i=k$, producing $\mu_k(u)$. If a band contains fewer than 25 training designs, its expert prediction is replaced by $\mu_g$. Each fitted expert has its own input scaling and covariance parameters.

The optional gate uses both $u$ and global predictive diagnostics. Write $\widetilde s=\max(s_g,10^{-6})$ and $a_j=(\mu_g-c_j)/\widetilde s$. Gaussian band probabilities are
\begin{equation}
 \pi_1=\Phi(-a_1),\qquad
 \pi_2=\Phi(-a_2)-\Phi(-a_1),\qquad
 \pi_3=1-\Phi(-a_2),
 \label{eq:bandprob}
\end{equation}
with numerical clipping and normalization. They enter the gate as predictive features. The diagnostic vector is
\begin{equation}
 d_g=(\mu_g,\widetilde s,a_1,a_2,|a_1|,|a_2|,
              \pi_1,\pi_2,\pi_3,H(\pi),\pi_1+\pi_3),
 \label{eq:diagnostics}
\end{equation}
where $H(\pi)=-\sum_k\pi_k\log\pi_k$. Its 11 components summarize response level, predictive scale, and proximity to band boundaries. 

Training diagnostics are obtained from two-fold, response-stratified cross-fitted global GPs. This split acts on the already assembled $u_i^{\mathrm{cf}}$: these two diagnostic folds reuse the previously fitted state representation. At prediction, diagnostics come from the full-training global GP.

After standardization of $[u,d_g]$, a class-balanced multinomial logistic model with inverse regularization parameter $C=0.7$ and a 2,000-iteration limit gives
\begin{equation}
 w_k(u,d_g)=\frac{\exp(\beta_k^T[u,d_g]+b_k)}
 {\sum_{j=1}^{3}\exp(\beta_j^T[u,d_g]+b_j)},
 \qquad w_k\ge0,\quad\sum_k w_k=1.
 \label{eq:gate}
\end{equation}
The regression output is
\begin{equation}
 \widehat f(x)=(1-\lambda)\mu_g(u(x))
       +\lambda\sum_{k=1}^{3}w_k(u(x),d_g(x))\mu_k(u(x)),
 \qquad \lambda=\tfrac12.
 \label{eq:predictor}
\end{equation}
The global path is retained even when one expert dominates the gate, and $\lambda$ is fixed in all comparisons. Taking $\lambda=0$ recovers the core global predictor. All expert comparisons use matched representations and training splits.

\section{Error propagation}\label{sec:analysis}
The analysis describes perturbations of a fixed fitted predictor. It separates state-estimation error from the additional routing error of the optional mixture.

\subsection{State-estimation error}
Write the fitted response map as
\begin{equation}
 F(b,z)=(1-\lambda)\mu_g(b,z)+\lambda\sum_k w_k(b,z)\mu_k(b,z),
 \label{eq:F}
\end{equation}
where $\lambda=0$ gives the core GP and $w$ includes its dependence on global predictive diagnostics.
\begin{proposition}\label{prop:state}
On a state domain containing $z$ and $\widehat z$, let $\mu_g$ and the expert means be $L_g$- and $L_e$-Lipschitz in $z$. Suppose $\|w(b,z_1)-w(b,z_2)\|_1\le L_w\|z_1-z_2\|_2$ and the expert span is at most $\overline D$. Then
\begin{align}
 |F(b,\widehat z)-F(b,z)|&\le C\|\widehat z-z\|_2,\label{eq:statebound}\\
 C&=(1-\lambda)L_g+\lambda L_e+\lambda\overline D L_w/2.\label{eq:stateconstant}
\end{align}
If these constants hold uniformly over the evaluation distribution and the second moments are finite, then
\begin{equation}
 \sqrt{R\{F(b,\widehat z)\}}\le\sqrt{R\{F(b,z)\}}+C\sqrt{\E\|\widehat z-z\|_2^2},
 \quad R(v)=\E[(v-y)^2].\label{eq:staterisk}
\end{equation}
\end{proposition}
\begin{proof}
Abbreviate evaluations at $\widehat z$ by primes. Decompose the mixture difference as $\sum_k w'_k(\mu'_k-\mu_k)+\sum_k(w'_k-w_k)\mu_k$. The first sum is bounded by $L_e\|\widehat z-z\|_2$. In the second, the weight differences sum to zero; subtracting the midpoint of the largest and smallest means bounds it by $\overline D\|w'-w\|_1/2$. Adding the global term proves \eqref{eq:statebound}, and the $L^2$ triangle inequality gives \eqref{eq:staterisk}.
\end{proof}
For the core GP, $C=L_g$. State accuracy and response sensitivity jointly control this conditional perturbation. The reference $F(b,z)$ uses the same fitted map; an exact-state model retrained on recorded labels is a different comparison. The thermal experiment evaluates the fixed-map substitution without estimating a uniform Lipschitz constant.

\subsection{Optional routing and regression interaction}
At a fixed representation, let $r=r(y)$ be the observed response band and define
\begin{equation}
 p=(1-\lambda)\mu_g+\lambda\sum_k w_k\mu_k,\qquad
 q=(1-\lambda)\mu_g+\lambda\mu_r.\label{eq:routingreference}
\end{equation}
For leakage $L=1-w_r$ and expert span $D=\max_{j,k}|\mu_j-\mu_k|$,
\begin{equation}
 |p-q|\le\lambda LD,\qquad
 \sqrt{R(p)}\le\sqrt{R(q)}+\lambda\sqrt{\E[L^2D^2]}.\label{eq:routingbound}
\end{equation}
The first inequality follows from $p-q=\lambda\sum_{k\ne r}w_k(\mu_k-\mu_r)$; the second again uses the $L^2$ triangle inequality. With $\delta=p-q$, the exact identity
\begin{equation}
 R(p)=R(q)+2\E[(q-y)\delta]+\E[\delta^2]\label{eq:riskidentity}
\end{equation}
retains a cross term of either sign. An assigned expert can itself have substantial regression error, so improved routing need not improve response accuracy. The numerical comparisons therefore evaluate expert utility directly at matched state inputs; the software tests separately verify the identity and pointwise bound.

\section{Implementation and computational cost}\label{sec:implementation}
The implementation separates the statistical construction from the simulator that supplies its labels. The \texttt{pssgp} Python package accepts a design matrix, an auxiliary-state matrix, and a response vector through \texttt{PhysicalStateSurrogate.fit(X,Z,y)}. Once fitted, \texttt{predict(X)} returns global and mixture predictions using only the design matrix. The numbers of design variables and states, the response-band boundaries, and the state transform are configurable. A DataFrame preserves the exact column schema, whereas array input is checked for its fitted dimension. Problem adapters define the units and transformations used in the reported experiments. A minimal use of the core predictor is
\begin{center}\begin{minipage}{.92\linewidth}\small
\begin{verbatim}
from pssgp import PhysicalStateSurrogate

# X: design table; Z: training-only states; y: scalar response
model = PhysicalStateSurrogate(backend="gp", folds=5, seed=42)
model.fit(X, Z, y)                   # cross-fitted encoder + global GP
y_new = model.predict(X_new)["gp"]   # design-only prediction
\end{verbatim}
\end{minipage}\end{center}
With \texttt{backend="moe"}, \texttt{predict} also returns the optional mixture prediction under key \texttt{"moe"}; \texttt{representation(X)} returns the augmented input $u(x)$, and a deterministic descriptor map $\phi$ can be supplied through \texttt{design\_features}.

The package uses NumPy, SciPy, and scikit-learn~\cite{harris2020,virtanen2020,pedregosa2011}; the publication figures use Matplotlib~\cite{hunter2007}. Empirical panels are rendered deterministically from the archived data with recorded source hashes. Encoder folds can run in independent processes, with one linear-algebra thread per worker to avoid oversubscription. A fresh thermal example generates a small database, trains the model, and saves predictions without an external hydrodynamics code. Regenerating ICF trajectories requires the physical simulator and material tables; fitting starts from the released tables.

For dense GPs, a covariance factorization on $n$ cases costs $O(n^3)$ and stores $O(n^2)$ entries~\cite{rasmussen2006}. With a fixed number of states and five encoder folds, the encoder stage remains cubic in $N$, with a larger constant than a single GP fit. The core response stage adds one global fit. The optional extension adds two diagnostic fits and expert costs proportional to $\sum_kN_k^3$, apart from repeated covariance evaluations during hyperparameter optimization. A query requires state and response means and a global variance; the latter involves a dense triangular solve and can cost $O(N^2)$. These costs restrict dense fitting to moderate databases.

Inducing-variable and stochastic variational GP methods offer alternatives for larger training sets~\cite{titsias2009,hensman2013}. The present implementation keeps dense covariance factorizations so that packaging can be checked directly against the accepted scientific calculations. For the CH Burn configuration with $N=1,000$ and seed 42, the general package reproduces cross-fitted states, test states, global predictions, and mixture predictions with zero maximum difference in the recorded Linux environment, against a fixed acceptance tolerance of $10^{-8}$. Seven unit tests and the independent thermal example also pass. The check covers one complete configuration. Because hyperparameter optimization proved sensitive to array layout, the package preserves the archived order of operations.

\begin{table}[tbp]
\centering\small
\caption{Observed wall-clock costs. Inference times are medians of seven repeated calls for the fitted CH package; the batch time is for the entire batch. Physical calculations and model timings refer to their distinct recorded environments and workloads.}\label{tab:cost}
\begin{tabularx}{\linewidth}{@{}l r X@{}}
\toprule
Operation & Time (s) & Measurement context\\
\midrule
CH on or off invocation & \result{ch_paired_wall_median} & Median of 3,974 archived calls; Windows/Intel Fortran, up to 12 independent workers; includes output processing.\\
CH surrogate fit & \result{package_fit_seconds} & Linux, $N=1,000$; five encoder-fold workers, single-thread linear algebra.\\
CH prediction, one design & \result{latency_1} & Encoder and both response outputs; one thread.\\
CH prediction, 987 designs & \result{latency_987} & The same model and prediction interface.\\
\bottomrule
\end{tabularx}
\end{table}

Table~\ref{tab:cost} reports the measured costs. The thermal ODE takes only \result{thermal_solver_ms}~ms per design, so that benchmark tests the information pathway rather than cost. For an existing training database, a simple cost accounting gives the break-even query count
\begin{equation}
 Q_* = \frac{C_{\mathrm{fit}}+C_{\mathrm{additional\ data}}}{C_{\mathrm{pair}}-C_{\mathrm{prediction}}},
 \qquad C_{\mathrm{pair}}>C_{\mathrm{prediction}},\label{eq:breakeven}
\end{equation}
where one direct amplification query requires both on and off calculations. Taking the existing database as a sunk cost, approximating $C_{\mathrm{pair}}$ by twice the archived median invocation time, and inserting the recorded complete-package fit and single-query times gives an illustrative threshold of \result{illustrative_break_even_paired} queries after rounding up. The timings come from different platforms, so this is a cost-accounting example, not a measured application speedup. New data generation, validation, and retraining increase the threshold. The experiments establish accuracy for scalar responses of one-dimensional calculations; extension to costlier multidimensional models would require its own physical-state definitions and validation.

\section{Test problems and evaluation protocol}\label{sec:protocol}
The principal test uses a recorded ICF database; a finite-fuel thermal problem and a newly generated Be/Cocktail database test the same procedure under different conditions.

\subsection{One-dimensional implosion database}
The CH database contains 1,987 designs generated with a one-dimensional spherical MULTI-IFE workflow~\cite{ramis2016}. Twenty design controls describe the imposed radiation drive and the target's layer geometry and density scales. Eleven drive coordinates control four pulse-amplitude segments, the radiation-temperature normalization, total duration, a time-shape exponent, the number of tabulated points assigned to the foot, and three relative pulse-interval redistributions. The remaining coordinates scale four layer thicknesses and five densities. The relative interval changes preserve the duration set by the global time scale. The released input table gives the coordinate order, bounds, and realized values. Every paired calculation uses 200 Lagrangian cells, distributed as 40, 60, 25, and 75 cells across the four layers; this resolution was verified against all 1,987 archived input files.

The label CH denotes the historical capsule family; its archived ablator setup uses a carbon material table. Each design has paired on and off calculations, their neutron yields, and trajectory diagnostics; disabling alpha deposition retains the nuclear-reaction diagnostics. The analysis uses a corrected rerun summary whose design identifiers link the inputs, paired yields, and state readouts. The original executable and runtime records were not retained, so this specified numerical database, rather than the historical runs, is the object of study.

The database spans weak and strongly amplified responses (Figure~\ref{fig:physics}a). The same case identifiers provide Peak and Burn supervision, so changing the state set does not alter the target response. Areal density is averaged with nonnegative neutron-production increments between stored times. The temperature and pressure columns instead use local burn-weighting proxies: each cell/interval receives a nonnegative weight proportional to $\rho^2 f_T^2\mathcal R(T_i)m_A\Delta t$, where $f_T$ is the recorded tritium fraction, $\mathcal R$ is the implemented DT-reactivity shape, and $m_A=\rho|\Delta r|$ is an areal-mass proxy. Adjacent-time values are midpoint averaged. The extraction functions and unit conversions accompany the data.

A small numerical sensitivity audit retains six designs, two from each response band, and repeats both feedback conditions on 200, 400, 800, and 1,600 cells. Layer interfaces and mesh grading are scaled consistently while the time controls are held fixed. All 48 calculations reach the prescribed endpoint. Their differences are reported in Section~\ref{sec:grid}.

\subsection{Finite-fuel thermal response}
The supporting problem couples an Arrhenius rate to heat release and cooling. Let $c(t)$ denote reacted fraction and $T_0$ the initial temperature. For a two-unit time interval,
\begin{align}
 \dot T&=S A\exp(-E/T)(1-c)-a(T-T_0),\label{eq:thermalT}\\
 \dot c&=A\exp(-E/T)(1-c),\qquad T(0)=T_0,\quad c(0)=0.
 \label{eq:thermalc}
\end{align}
Here $S>0$ sets the temperature rise available from complete conversion and $a>0$ is a cooling rate. The finite-fuel factor $1-c$ limits the total release.

Eight independent normalized design coordinates are mapped affinely to the
ordered variables $(T_0,P,\varphi,R,E,\log_{10}A,k_w,q)$. The corresponding
lower and upper bounds are
\begin{align*}
 \ell&=(550,0.5,0.2,0.02,1.4\times10^4,7,5,0.6),\\
 u&=(750,3,1,0.10,2.2\times10^4,9.5,60,1.4).
\end{align*}
In the benchmark's fixed unit convention,
\begin{equation}
 c_0=\frac{40P}{T_0/600},\qquad
 S=\frac{2.5\times10^6 q c_0\varphi}{1200},\qquad
 a=\frac{3k_w}{1200R}.
\end{equation}
The reference evolution freezes the rate at $k_0=A\exp(-E/T_0)$ while retaining fuel depletion. It has the closed-form solution
\begin{align}
 c_{\mathrm{ref}}(t)&=1-\exp(-k_0t),\label{eq:thermalrefc}\\
 T_{\mathrm{ref}}(t)&=T_0+S k_0\frac{\exp(-k_0t)-\exp(-at)}{a-k_0}.
 \label{eq:thermalrefT}
\end{align}
For $a=k_0$, the second expression is continued as $T_0+S k_0t\exp(-k_0t)$. Its peak on $[0,2]$ is evaluated at the clipped stationary time $\log(a/k_0)/(a-k_0)$, with limit $1/a$. The response and auxiliary states are
\begin{align}
 y_{\mathrm{th}}&=\log_{10}\frac{c(2)}{c_{\mathrm{ref}}(2)},\label{eq:thermaltarget}\\
 z_{\mathrm{th}}&=\bigl(\log_{10}T_{\mathrm{ref,peak}},\,
 \log_{10}c_{\mathrm{ref}}(2),\,\log_{10}(Sk_0)\bigr).
 \label{eq:thermalstates}
\end{align}
The implementation floors the conversion quantities at $10^{-12}$ and $Sk_0$ at $10^{-9}$ before taking logarithms. The state vector includes the reference's final reacted fraction, which also forms the target denominator and is therefore a particularly informative auxiliary target. The base representation adds $\log_{10}[Sk_0/(aT_0)]$, $\log_{10}(S/T_0)$, and $T_0/E$ to the eight normalized coordinates. The complete learned-state comparison withholds all three analytic states at prediction, learns them by five-fold cross-fitting from the design descriptors, and supplies the estimates to both response backends. Separate fits with the exact analytic states provide an information-condition reference.

The full response is integrated by LSODA with relative tolerance $10^{-6}$, absolute tolerance $10^{-9}$, and maximum step $0.02$. A numerical guard terminates integration if $T$ reaches $100T_0$, assigning complete conversion. This guard is activated for \result{thermal_guard_count} of \result{thermal_n_thermal} designs and is part of the recorded target-generation protocol. A stored audit of \result{thermal_n_audit} cases checks tighter integration, analytic reference evaluation, and continuation beyond the guard. The maximum change in the log-response under tighter tolerances is $\result{thermal_max_target_tolerance_difference}$; the maximum guard-continuation difference in the audited cases is $\result{thermal_max_guard_continuation_difference}$.

\subsection{Reconstructed Be/Cocktail database}
The second ICF target family varies peak and foot drive amplitudes and the thicknesses of DT gas, DT ice, Cocktail, and Be layers. A fixed random sample of 3,000 designs is drawn without using responses from the historical design pool. The archived input generator and ten material tables are recovered, and all formatted inputs match the selected historical coordinates. Because the historical production executable could not be identified, all targets and off-state diagnostics are recomputed with a single recorded source build.

The 200-cell spherical setup retains the prescribed 20~ns endpoint. Disabling alpha deposition leaves neutron production and burn diagnostics active. All 6,000 calls complete and all off trajectories pass the zero-alpha-energy check. The resulting library contains 2,067 low-band, 653 intermediate-band, and 280 high-band designs under \eqref{eq:bands}. Burn states use the same diagnostic definitions as CH. The design representation has 17 coordinates: six raw controls, six logarithmic controls, and five drive or geometry ratios. The learned-state representation has 20 coordinates.

Random holdout fixes 600 test designs. The second split reserves the 600 largest physical outer radii, calculated as the sum of the four layer thicknesses. This is a physical-radius split, unlike the normalized design-space radius used for CH. Five seeds select nested training sets of 250 and 1,000 from the remaining 2,400 designs. The validated package supplies learned-state GP and GP--MoE predictions; the archived backend supplies the design-only controls with the same kernel and gate settings. No hyperparameters or state definitions are selected using these Be test results.

\subsection{Training and evaluation splits}
CH random holdout is stratified using amplification boundaries $1.5,2,5,20,100$. At $N=1,000$, the remaining 987 designs form the test set. The second split extrapolates in a design-space radius,
\begin{equation}
 r_x=\left[\frac1{20}\sum_{j=1}^{20}
 \left(\frac{x_j-m_j}{s_j}\right)^2\right]^{1/2},
 \label{eq:radius}
\end{equation}
where $m_j$ and $s_j$ are the median and interquartile range of the database's design coordinates. The central 60\% supplies the training pool; the outer 20\%, comprising 398 designs, supplies the test set. The intervening designs are excluded. The split uses covariates of the whole database but no test responses. The CH figures use ``Radius OOD'' to denote this design-space split.

The main $N=1,000$ comparisons use ten paired seeds per split. Learning curves use $N=125,250,375,500,750,1,000$. Smaller training sets are nested prefixes of a seeded permutation of the parent training set; the parent test set is held fixed. State encoders are rebuilt at each budget. The six-budget curves use ten common seeds on random holdout and seven on radius OOD, giving 102 complete split--seed--budget configurations. The remaining three OOD seeds are excluded from the curves because their parent $N=1,000$ fits failed the $10^{-8}$ reproduction check; their identifiers and audit records are released. The additional budgets $125,375,750$ were fixed before their outputs were examined. The separate primary $N=1,000$ comparison uses all ten saved seeds on both splits.

The thermal database has \result{thermal_n_thermal} designs. Random holdout uses a fixed pool of 2,000 and a test set of 1,000. A normalized radius from the center of the eight-dimensional design cube defines a central-60\% training pool and an outer-20\% test set for the second split. Five seeds are used at each of $N=320$ and $N=1,000$. The two budgets use separately selected training samples; $N=1,000$ was a subsequent check of the initial smaller-budget result. The GP--MoE follows the same global--local blend and minimum expert size, with unstratified two-fold global diagnostics and an 80-iteration limit in GP hyperparameter optimization. Its encoder follows the same cross-fitted information path. All 20 configurations are evaluated against reference fits rebuilt in one Linux environment.

\subsection{Supervision controls and physical interpretation}\label{sec:controls}
At $N=1,000$, the CH attribution experiment keeps the encoder and response backends fixed and changes only auxiliary supervision. The random-projection control uses three unit-norm Gaussian projections of standardized design descriptors; its coordinate system is defined from the outer training designs. The shuffled-state control jointly permutes the three Burn columns across training cases, preserving their marginal and joint distribution while breaking their case correspondence. Both controls retain the same three-output encoder and cross-fitting procedure. Recorded Burn states supplied directly at training and test form an additional information-condition reference that is unavailable at design-only deployment.

The submission-time protocol names Burn as the primary state set and evaluates every split/backend cell. With $G$ defined by \eqref{eq:improvement}, its strong-support criterion requires positive Burn gain and each negative control to retain at most one quarter of that gain in all four cells. Seedwise differences are also summarized using 10,000 paired bootstrap resamples.

A further state set is motivated by the no-alpha ignition parameter of Betti et al.~\cite{betti2015}. Using the available recorded diagnostics, we form
\begin{equation}
 \widetilde\chi_{\mathrm{off}}=(\rho R_{\mathrm{off}})^{0.61}
 \left(0.24\,\frac{Y_{\mathrm{off}}/10^{16}}{M_{0,\mathrm{DT}}[\mathrm{mg}]}\right)^{0.34}.
 \label{eq:betti}
\end{equation}
Here $\rho R$ is in $\mathrm{g\,cm^{-2}}$. The published parameter uses an unablated-fuel mass and a shell diagnostic; we substitute initial total DT mass and burn-weighted total-DT areal density. Equation~\eqref{eq:betti} therefore defines a Betti-inspired proxy. The three supervised log coordinates are areal density, yield per initial fuel mass, and this proxy. The third is an affine combination of the first two, so the set carries at most two independent coordinates. The same off yield appears in \eqref{eq:response} and \eqref{eq:betti}; this shared normalization contributes to their association.
The coefficient 0.24 follows Eq.~(11) of Betti et al.~\cite{betti2015}; any fixed positive coefficient only shifts the logged coordinate, and training-fold standardization removes that shift.

\subsection{Baselines and error measures}
The principal architecture control is a global GP with exactly the same response inputs. A second comparison supplies the predicted Burn representation to Extra Trees, histogram gradient boosting, and a multilayer perceptron (MLP). Each family has 12 candidate configurations selected by three-fold inner validation of the outer training set, to separate model selection from outer evaluation~\cite{cawley2010}. The state encoder, including its cross-fitted training states and full-fold prediction model, is rebuilt inside each inner split. This prevents the inner validation states from being generated by an encoder trained on their own physical labels.

Extra Trees uses 500 trees, minimum leaf sizes in $\{1,2,4,8\}$ and maximum feature fractions in $\{0.5,0.8,1\}$. Boosted trees use 400 iterations, learning rates in $\{0.03,0.1\}$, maximum leaf counts in $\{7,15,31\}$, and $L^2$ penalties in $\{0,1\}$. MLP candidates combine hidden widths $(64)$, $(128,64)$, or $(128,128)$ with penalties $0.001$ or $0.1$ and either Adam or L-BFGS optimization. Their iteration limit is 2,000, the L-BFGS function-evaluation limit is 30,000, and the Adam learning rate is $0.001$; early stopping is disabled. Preprocessing is fitted within each training split. The GP settings remain those of Section~\ref{sec:method}.

The main metric is test MSE of the base-ten logarithmic response. Reported error bars are sample standard deviations across split seeds. For paired methods A and B, a positive percentage reduction means
\begin{equation}
 I_{A\to B}=100\left(1-\frac{\overline{\mathrm{MSE}}_B}
 {\overline{\mathrm{MSE}}_A}\right),
 \label{eq:improvement}
\end{equation}
where the bars average MSE over the same seeds. Paired win counts compare the two MSE values within each seed. For regional diagnosis, low amplification is $\Yamp<2$, the cliff neighborhood is $1.5\le\Yamp\le5$, the transition region is $2\le\Yamp\le20$, and the high tail is $\Yamp>100$.

\section{Results}\label{sec:results}
The comparisons first test whether physical supervision improves predictions after the reference solve has been removed from inference. They then examine the state representation, the extra contribution of experts, and portability beyond the principal database.

\subsection{Physical supervision survives design-only inference}
Table~\ref{tab:supervision} and Figure~\ref{fig:controls} show the complete CH supervision experiment. Learned Burn states reduce mean MSE relative to the same-backend Design predictor by \result{e1_random_gp_burn3_gain}\% for GP and \result{e1_random_moe_burn3_gain}\% for GP--MoE on random holdout. On design-radius extrapolation the reductions are \result{e1_radius_ood_gp_burn3_gain}\% and \result{e1_radius_ood_moe_burn3_gain}\%. Each comparison favors Burn in all ten paired seeds. Because test states are predicted, these gains require no reference calculation at deployment.

\begin{table}[tbp]\centering\small
\caption{CH supervision controls at $N=1,000$: mean test MSE of the log amplification over ten seeds per split. Recorded Burn is a separate condition with simulated states supplied at prediction.}\label{tab:supervision}
\begin{tabular}{@{}lrrrr@{}}
\toprule
Supervision & Random GP & Random MoE & OOD GP & OOD MoE\\
\midrule
Design & \result{e1_random_gp_design_mse} & \result{e1_random_moe_design_mse} & \result{e1_radius_ood_gp_design_mse} & \result{e1_radius_ood_moe_design_mse}\\
Peak & \result{e1_random_gp_current3_mse} & \result{e1_random_moe_current3_mse} & \result{e1_radius_ood_gp_current3_mse} & \result{e1_radius_ood_moe_current3_mse}\\
Burn & \result{e1_random_gp_burn3_mse} & \result{e1_random_moe_burn3_mse} & \result{e1_radius_ood_gp_burn3_mse} & \result{e1_radius_ood_moe_burn3_mse}\\
Random projection & \result{e1_random_gp_rand3_mse} & \result{e1_random_moe_rand3_mse} & \result{e1_radius_ood_gp_rand3_mse} & \result{e1_radius_ood_moe_rand3_mse}\\
Shuffled Burn & \result{e1_random_gp_shuffle3_mse} & \result{e1_random_moe_shuffle3_mse} & \result{e1_radius_ood_gp_shuffle3_mse} & \result{e1_radius_ood_moe_shuffle3_mse}\\
Betti-inspired & \result{e1_random_gp_betti3_mse} & \result{e1_random_moe_betti3_mse} & \result{e1_radius_ood_gp_betti3_mse} & \result{e1_radius_ood_moe_betti3_mse}\\
Recorded Burn & \result{e1_random_gp_oracle_burn_mse} & \result{e1_random_moe_oracle_burn_mse} & \result{e1_radius_ood_gp_oracle_burn_mse} & \result{e1_radius_ood_moe_oracle_burn_mse}\\
\bottomrule
\end{tabular}
\end{table}

\begin{figure}[tbp]
\centering
\includegraphics[width=.488\linewidth]{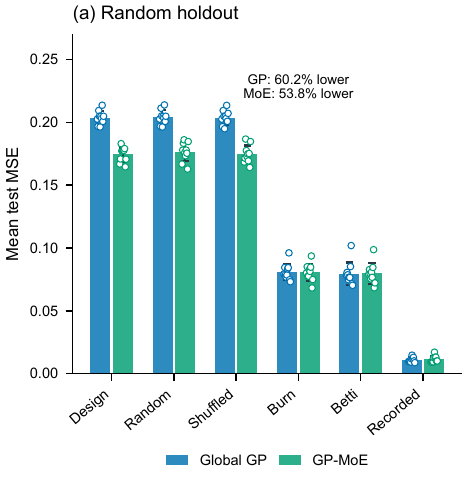}\hfill
\includegraphics[width=.488\linewidth]{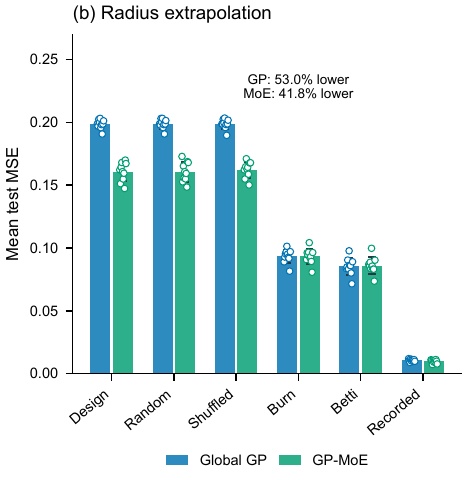}
\caption{CH supervision controls. Bars are mean MSE, whiskers are sample SD, and points show all ten seeds. Annotated reductions compare learned Burn with Design within the same backend. Random and Shuffled preserve the auxiliary pathway while changing its supervision; Recorded supplies true test states.}\label{fig:controls}
\end{figure}

Both negative controls have slightly higher mean error than Design in every cell. Burn beats both controls in all paired seeds, and the paired bootstrap intervals for their excess MSE exclude zero. The preregistered strong-support criterion is met in every cell. The improvement therefore depends on the physical case/state correspondence under the two specified controls.

Direct access to recorded Burn states gives substantially smaller errors than learned states (Table~\ref{tab:supervision}). The physical labels are therefore informative beyond what the finite-data encoder recovers. This comparison changes the training and test information, whereas Proposition~\ref{prop:state} holds the fitted map fixed.

\subsection{Learning curves and physical coordinates}
Figure~\ref{fig:learning} follows the learned representations across six nested training budgets. For GP--MoE, Burn reduces random-holdout MSE relative to Design by \result{dense_representation_gain_random_moe_125}\% at $N=125$ and \result{learn_random_moe_1000}\% at $N=1,000$; the OOD reductions are \result{dense_representation_gain_radius_ood_moe_125}\% and \result{learn_radius_ood_moe_1000}\%. The corresponding GP reductions are \result{dense_representation_gain_random_gp_125}\% and \result{learn_random_gp_1000}\% on random holdout, and \result{dense_representation_gain_radius_ood_gp_125}\% and \result{learn_radius_ood_gp_1000}\% on OOD. The added state supervision benefits both backends; the OOD curves use seven common seeds.

\begin{figure}[tbp]
\centering
\includegraphics[width=.488\linewidth]{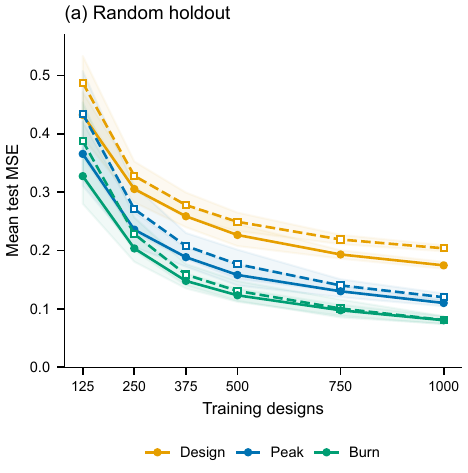}\hfill
\includegraphics[width=.488\linewidth]{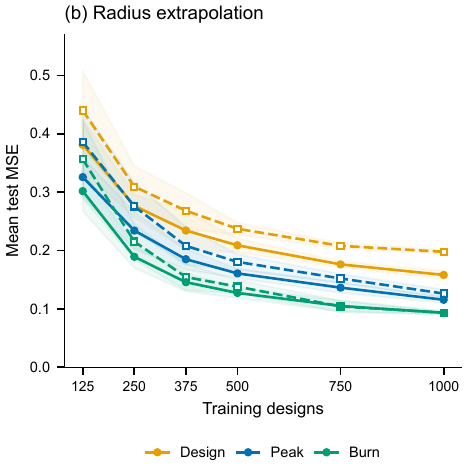}
\caption{CH learning curves over six training budgets. Dashed lines with open squares denote global GP; solid lines with circles denote GP--MoE. Colors identify Design, Peak, and Burn representations. Shading is sample SD. The same ten random-holdout seeds and seven complete OOD seeds enter every budget; inclusion criteria are specified in Section~\ref{sec:protocol}.}\label{fig:learning}
\end{figure}

The fixed-seed parity plots in Figure~\ref{fig:parity} illustrate the change for a global GP. Identical test identities, axis limits, and density scales are used. The improved representation narrows the broad response scatter without changing the deployment input or requiring local experts.

\begin{figure}[tbp]
\centering
\includegraphics[width=.488\linewidth]{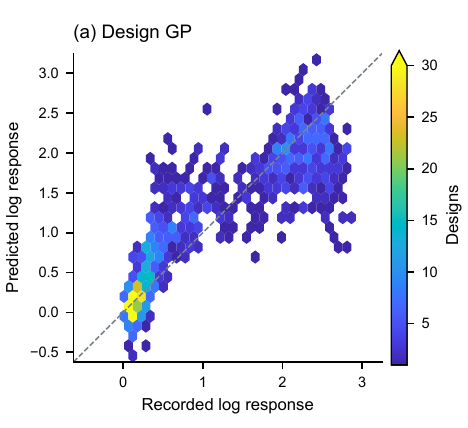}\hfill
\includegraphics[width=.488\linewidth]{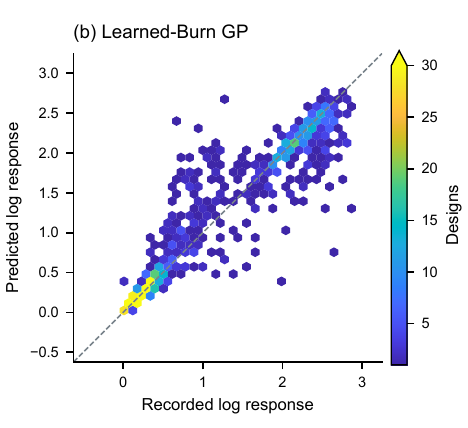}
\caption{Global-GP predictions for random-holdout seed 42 at $N=1,000$, with (a) Design and (b) learned Burn inputs. Both panels contain the same 987 test designs and use identical axis and color limits. Dashed lines indicate equality; colors count designs per hexagon.}\label{fig:parity}
\end{figure}

The physical interpretation of the state set is supported by recorded associations, with appropriate limits. Across the CH library, the Betti-inspired proxy has Spearman coefficient \result{assoc_betti_proxy} with the response (Figure~\ref{fig:physics}b); burn-weighted pressure has coefficient \result{assoc_off_burn_weighted_pstag_gbar}. The proxy, however, contains the off yield that also normalizes the target.

\begin{figure}[tbp]
\centering
\includegraphics[width=.488\linewidth]{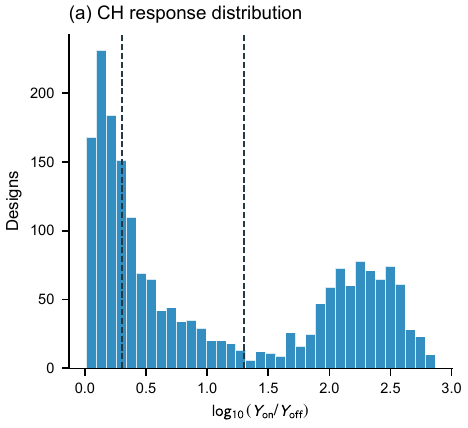}\hfill
\includegraphics[width=.488\linewidth]{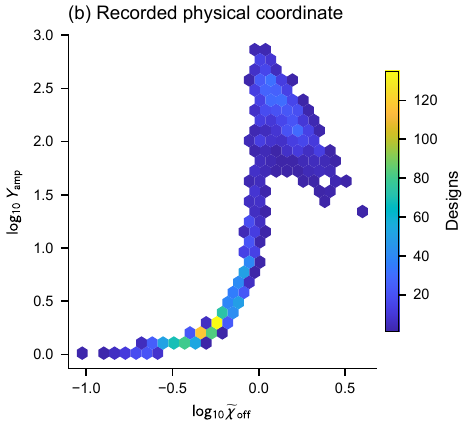}
\caption{CH database diagnostics: (a) response distribution, with thresholds 2 and 20; (b) association with the recorded Betti-inspired proxy of \eqref{eq:betti}. This whole-library diagnostic uses true reference quantities, includes a shared off-yield normalization, and describes the recorded physical coordinates.}\label{fig:physics}
\end{figure}

Learning the Betti-inspired coordinates instead of Burn gives additional mean reductions of \result{e1_random_gp_betti_vs_burn}\% and \result{e1_random_moe_betti_vs_burn}\% on random holdout, and \result{e1_radius_ood_gp_betti_vs_burn}\% and \result{e1_radius_ood_moe_betti_vs_burn}\% on OOD, for GP and GP--MoE respectively. The random-holdout paired intervals cross zero; both OOD comparisons favor the proxy in all ten seeds. A subsequently triggered, fixed-protocol six-budget comparison contains 102 configurations and finds a budget-dependent advantage, including reversals at small budgets. Burn remains the primary state set.

Figure~\ref{fig:states} measures what the encoder actually recovers. Component errors are normalized by the finite log-state population SD in the outer training set. The mean NRMSEs for $(\rho R,T_i,P)$ are $(\result{state_random_rhor},\result{state_random_temperature},\result{state_random_pressure})$ on random holdout and $(\result{state_radius_ood_rhor},\result{state_radius_ood_temperature},\result{state_radius_ood_pressure})$ on OOD. Appreciable state error remains, especially under extrapolation. The response improvement is therefore obtained with imperfect states.

\begin{figure}[tbp]
\centering
\includegraphics[width=.488\linewidth]{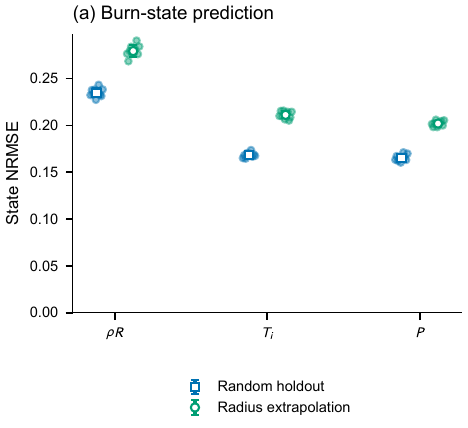}\hfill
\includegraphics[width=.488\linewidth]{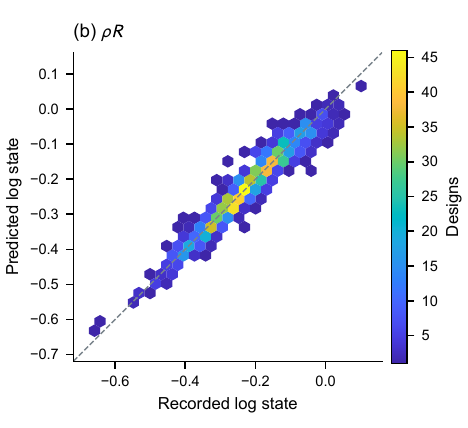}

\includegraphics[width=.488\linewidth]{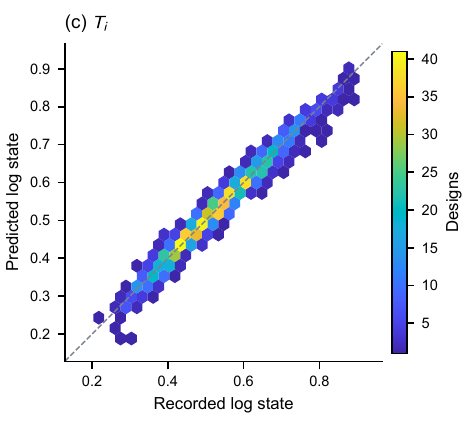}\hfill
\includegraphics[width=.488\linewidth]{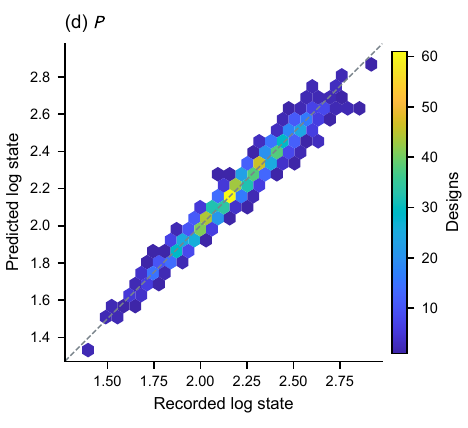}

\caption{CH Burn-state recovery. (a) NRMSE for ten seeds per split, normalized with training-state SD; points, hollow markers, and intervals denote seeds, means, and sample SD. (b--d) Predicted versus recorded log states for fixed random seed 42. Only finite recorded labels are plotted. Each state has its own axes and density scale; all are encoder predictions on held-out designs.}\label{fig:states}
\end{figure}

Peak-to-Burn gains vary across the response range. Figure~\ref{fig:regions} shows the regional reductions. In the OOD high tail, $\Yamp>100$, mean MSE increases by \result{tail_radius_ood_gp}\% for GP and \result{tail_radius_ood_moe}\% for GP--MoE. Random-holdout high-tail errors instead decrease by \result{tail_random_gp}\% and \result{tail_random_moe}\%. The response benefit and the state-estimation problem must therefore be assessed in the region relevant to an application.

\begin{figure}[tbp]
\centering
\includegraphics[width=.488\linewidth]{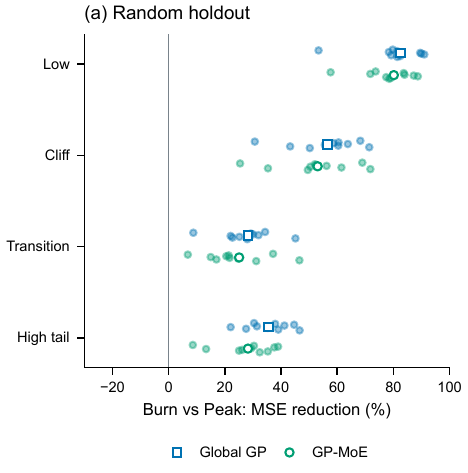}\hfill
\includegraphics[width=.488\linewidth]{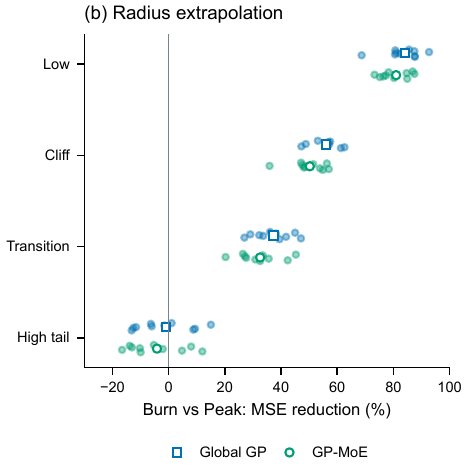}
\caption{CH regional Peak-to-Burn changes. Large markers are ratios of mean MSEs over ten seeds; small points show seedwise percentage changes. Positive values favor Burn. Low, cliff, transition, and high-tail subsets are defined in Section~\ref{sec:protocol} and overlap. The OOD high-tail comparison favors Peak.}\label{fig:regions}
\end{figure}

\subsection{The additional value of experts}
Matched-input comparisons isolate the complete expert-and-gate augmentation from state information. With Burn, the random-holdout gain decreases from \result{dense_expert_gain_random_burn3_125}\% at $N=125$ to \result{expert_random_burn3_1000}\% at $N=1,000$ (Figure~\ref{fig:architecturegain}a). The complete ten-seed $N=1,000$ results give only \result{primary_expert_random}\% gain on random holdout and \result{primary_expert_radius_ood}\% on OOD. At this budget the global Burn predictor is a competitive, simpler choice.

This trend is representation dependent. On Design inputs the random-holdout expert gain increases from \result{expert_random_design_250}\% at $N=250$ to \result{expert_random_design_1000}\% at $N=1,000$. Peak retains an extra benefit at the largest budget. Training size and representation therefore change expert utility along distinct axes. The conditional routing term in \eqref{eq:routingbound} and the signed cross term in \eqref{eq:riskidentity} allow either outcome.

The same-input comparison with internally tuned regressors is shown in Figure~\ref{fig:architecturegain}c,d. GP--MoE reduces mean MSE relative to Extra Trees, the strongest tested non-GP family, by \result{baseline_random_moe_et}\% on random holdout and \result{baseline_radius_ood_moe_et}\% on OOD. The corresponding global-GP gains are \result{baseline_random_gp_et}\% and \result{baseline_radius_ood_gp_et}\%. Both GP backends beat each of the three tuned families in every paired seed on both splits. This supports their competitiveness under the stated tuning budget, while showing that most of this comparison does not require experts. Deep kernels and fully Bayesian nonstationary GPs remain untested alternatives.

\begin{figure}[tbp]
\centering
\includegraphics[width=.488\linewidth]{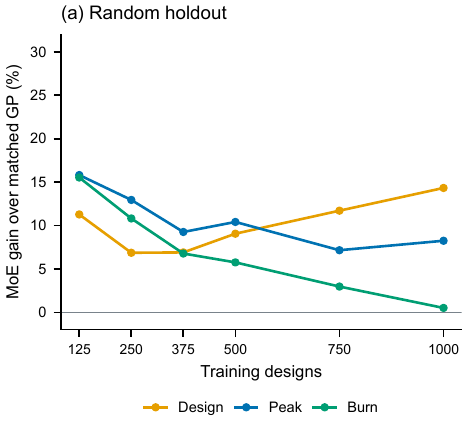}\hfill
\includegraphics[width=.488\linewidth]{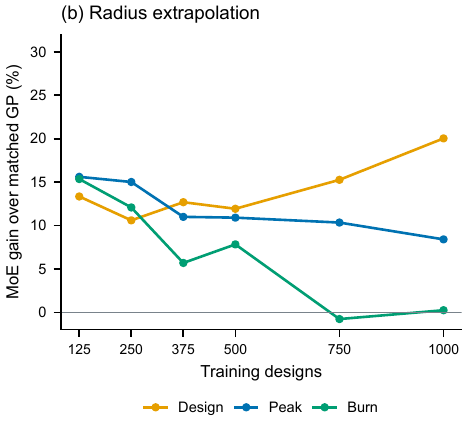}

\includegraphics[width=.488\linewidth]{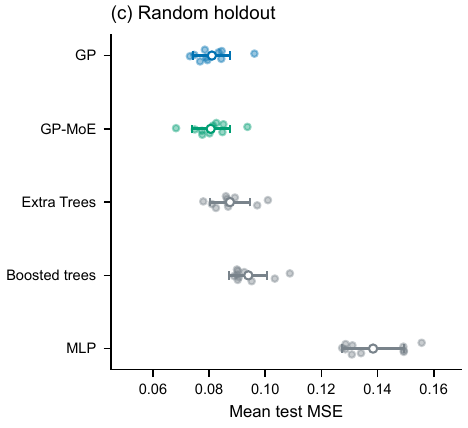}\hfill
\includegraphics[width=.488\linewidth]{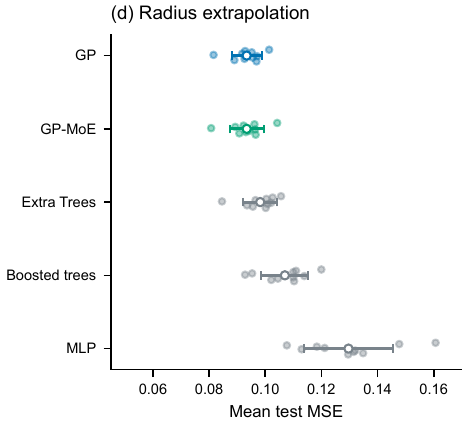}

\caption{Architecture and regression-model comparisons. (a,b) Additional GP--MoE gain over a global GP with identical inputs along the complete-case learning curves. (c,d) Ten-seed $N=1,000$ comparisons with internally tuned regressors, all using predicted Burn states; points show seeds, with mean and sample SD.}\label{fig:architecturegain}
\end{figure}

\subsection{Complete learned-state validation in the thermal model}
The learned-state thermal experiment tests the full training-to-inference pathway (Figure~\ref{fig:thermal}). Relative to Design, the learned GP reduces mean MSE by \result{e6_random_320_gp_gain}\% and \result{e6_random_1000_gp_gain}\% on random holdout at $N=320$ and $1,000$, and by \result{e6_radius_320_gp_gain}\% and \result{e6_radius_1000_gp_gain}\% on radius OOD. The corresponding GP--MoE reductions are \result{e6_random_320_moe_gain}\%, \result{e6_random_1000_moe_gain}\%, \result{e6_radius_320_moe_gain}\%, and \result{e6_radius_1000_moe_gain}\%. Mean improvements occur in all eight same-backend comparisons, although the largest-budget OOD comparisons win only four and three of five seeds for GP and GP--MoE.

The expert effect has a different budget dependence from CH. At $N=320$, its gains relative to the learned GP are \result{e6_random_320_expert}\% and \result{e6_radius_320_expert}\%, both unfavorable. At $N=1,000$ they become \result{e6_random_1000_expert}\% and \result{e6_radius_1000_expert}\%.

For the state-error analysis, each fitted predictor is additionally evaluated after replacing its estimated test states with the exact analytic states. Figure~\ref{fig:thermal}b plots the state error against the change of that same GP. This substitution matches the fixed-map reference in Proposition~\ref{prop:state}. Exact-state fits sometimes have higher error than learned-state fits, consistent with their different finite-sample training representations.

\begin{figure}[tbp]
\centering
\includegraphics[width=.488\linewidth]{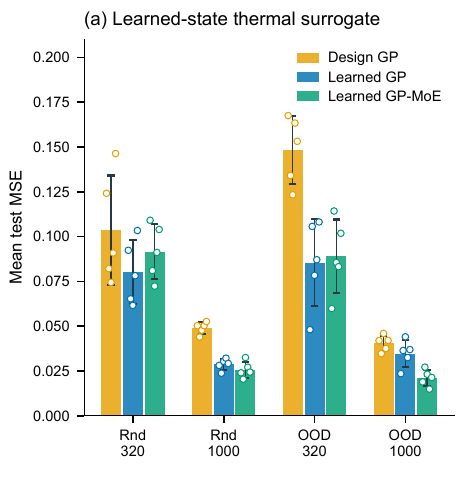}\hfill
\includegraphics[width=.488\linewidth]{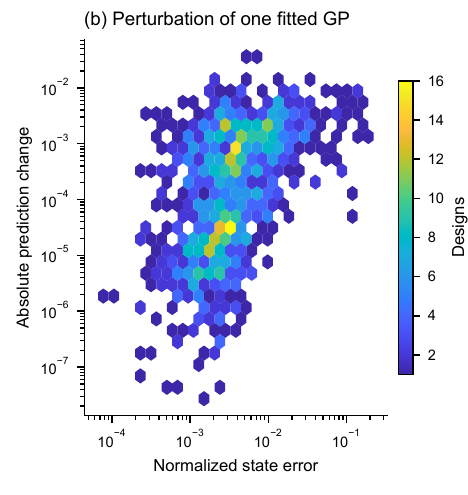}
\caption{Thermal learned-state validation. (a) Mean test MSE, sample SD, and all five seeds for selected backends; the full Design/Exact/Learned GP and GP--MoE results are retained in the numerical tables. Rnd and OOD denote the two fixed splits. (b) Fixed-model state substitution for random seed 11, $N=1,000$: normalized state error versus absolute change in the fitted GP prediction, on logarithmic axes.}\label{fig:thermal}
\end{figure}

\subsection{Be/Cocktail transfer and its limits}
The newly calculated Be/Cocktail test gives mixed evidence (Figure~\ref{fig:be}). On random holdout, learned Burn reduces GP MSE by \result{be_random_250_design_gp_burn_gp_gain_pct}\% at $N=250$ and \result{be_random_1000_design_gp_burn_gp_gain_pct}\% at $N=1,000$, winning four of five seeds at each budget. The paired bootstrap intervals for these gains exclude zero. Thus the physical-state pathway can help beyond the principal capsule family.

The physical-radius extrapolation test is less favorable. Its GP reductions are only \result{be_radius_ood_250_design_gp_burn_gp_gain_pct}\% and \result{be_radius_ood_1000_design_gp_burn_gp_gain_pct}\%, with intervals spanning zero. The learned mixture is worse than the matched learned GP at $N=1,000$ in every seed: the gain is \result{be_random_1000_burn_gp_burn_moe_gain_pct}\% on random holdout and \result{be_radius_ood_1000_burn_gp_burn_moe_gain_pct}\% on OOD. The latter comparison has a paired interval wholly below zero. The target-family change exposes a distinct extrapolation difficulty and a clear preference for the global backend at the larger budget.

\begin{figure}[tbp]
\centering
\includegraphics[width=.488\linewidth]{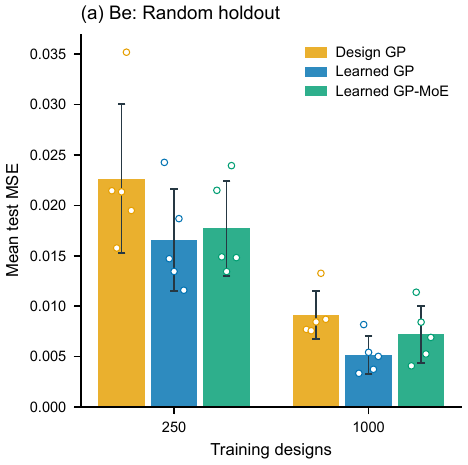}\hfill
\includegraphics[width=.488\linewidth]{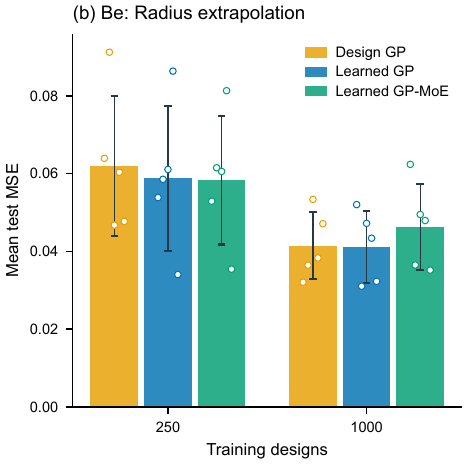}
\caption{New Be/Cocktail database: mean MSE, sample SD, and five seeds at each budget. Learned states are predicted from the 17-dimensional design representation. The OOD split reserves the largest physical outer radii, distinct from CH design-radius OOD. The complete matched Design-MoE controls and paired intervals are provided with the numerical record.}\label{fig:be}
\end{figure}

\subsection{Numerical sensitivity of the underlying ICF data}\label{sec:grid}
The training database uses 200 cells. Against the finest tested mesh of 1,600 cells, its six-design audit gives a maximum absolute log-amplification difference of \result{grid_200_log10_yamp_max_abs_difference} dex and a median of \result{grid_200_log10_yamp_median_abs_difference} dex (Table~\ref{tab:grid}). For scale, the random-holdout Burn GP has a root-mean-square error of \result{burn_random_rmse} dex, obtained by taking the square root of its mean seedwise MSE. The six selected mesh differences and the full-test predictive RMSE measure different errors and populations; their similar scale warns against interpreting statistical accuracy as continuum accuracy.

\begin{table}[tbp]\centering\small
\caption{Six-design mesh sensitivity against 1,600 cells. The 200-cell row is the database resolution. Log-response differences are absolute; the remaining four columns are maximum absolute relative differences. Time controls are unchanged, and the finest mesh is a comparison reference rather than an established continuum solution.}\label{tab:grid}
\begin{tabular}{@{}lrrrrrr@{}}
\toprule
Cells & $|\Delta y|_{\max}$ & $|\Delta y|_{\rm med}$ & $Y_{\rm off}$ (\%) & $\rho R$ (\%) & $T_i$ (\%) & $P$ (\%)\\
\midrule
200 & \result{grid_200_log10_yamp_max_abs_difference} & \result{grid_200_log10_yamp_median_abs_difference} & \result{grid_200_off_neutrons_max_abs_relative_pct} & \result{grid_200_off_burn_weighted_rhordt_max_abs_relative_pct} & \result{grid_200_off_burn_weighted_tion_kev_max_abs_relative_pct} & \result{grid_200_off_burn_weighted_pstag_gbar_max_abs_relative_pct}\\
400 & \result{grid_400_log10_yamp_max_abs_difference} & \result{grid_400_log10_yamp_median_abs_difference} & \result{grid_400_off_neutrons_max_abs_relative_pct} & \result{grid_400_off_burn_weighted_rhordt_max_abs_relative_pct} & \result{grid_400_off_burn_weighted_tion_kev_max_abs_relative_pct} & \result{grid_400_off_burn_weighted_pstag_gbar_max_abs_relative_pct}\\
800 & \result{grid_800_log10_yamp_max_abs_difference} & \result{grid_800_log10_yamp_median_abs_difference} & \result{grid_800_off_neutrons_max_abs_relative_pct} & \result{grid_800_off_burn_weighted_rhordt_max_abs_relative_pct} & \result{grid_800_off_burn_weighted_tion_kev_max_abs_relative_pct} & \result{grid_800_off_burn_weighted_pstag_gbar_max_abs_relative_pct}\\
\bottomrule
\end{tabular}
\end{table}

At the database resolution, maximum relative differences reach \result{grid_200_off_burn_weighted_tion_kev_max_abs_relative_pct}\% for off temperature and \result{grid_200_off_burn_weighted_pstag_gbar_max_abs_relative_pct}\% for pressure. Areal-density differences decrease from \result{grid_200_off_burn_weighted_rhordt_max_abs_relative_pct}\% to \result{grid_400_off_burn_weighted_rhordt_max_abs_relative_pct}\% and \result{grid_800_off_burn_weighted_rhordt_max_abs_relative_pct}\% under refinement. Off-yield differences likewise fall from \result{grid_200_off_neutrons_max_abs_relative_pct}\% to \result{grid_800_off_neutrons_max_abs_relative_pct}\%, but remain appreciable at 200 cells. Temperature and pressure are substantially more mesh sensitive over the tested range.

The Betti-inspired supervision uses areal density and off yield per initial fuel mass, without the recorded temperature or pressure columns. Its competitive prediction errors show that the observed supervision benefit can be obtained without those two diagnostics. They do not establish mesh independence: the off yield changes with resolution, and no cross-resolution surrogate test was performed. All learning comparisons therefore retain their specified database, while validation on physically converged labels remains a separate requirement.

\section{Discussion and conclusions}\label{sec:discussion}
The principal contribution is a deployable use of physical-state supervision. A reference calculation supplies named physical labels during training; cross-fitted estimates connect those labels to the information available at prediction. A retained design path avoids imposing a three-variable physical closure, and matched global and expert predictors allow representation and architecture to be assessed separately. The construction combines established statistical components in an explicit train-to-deployment architecture.

The supervision controls give the strongest attribution evidence. The CH improvements persist after replacing recorded test states by encoder predictions, and neither random projection nor shuffled physical labels reproduces them through the same pathway. This supports the value of case-aligned physical supervision for the tested database. Since the predicted states are deterministic functions of the design, the benefit lies in the finite-sample regression problem, through information learned from auxiliary labels.

The state-estimation stage leaves substantial headroom. On random holdout, the GP fitted with recorded Burn states has mean MSE \result{e1_random_gp_oracle_burn_mse}, compared with \result{e1_random_gp_burn3_mse} for learned states. This diagnostic gap includes both state-prediction error and the different inputs used to train the response fits; it is not the fixed-map perturbation of Proposition~\ref{prop:state}. Improving the encoder, especially in extrapolation and the high-response tail, is a more direct next step than adding expert complexity to an already informative representation.

Expert specialization has a narrower role. It helps several small-sample CH comparisons, contributes little at $N=1,000$ with Burn inputs, improves the larger thermal tests, and degrades the larger Be tests. These outcomes support choosing the representation first and validating the optional expert extension at matched inputs. The error analysis leaves encoder sensitivity, expert disagreement, and signed regression--routing interaction explicit, consistent with these different outcomes.

The evidence has physical and statistical limits. CH and Be share a one-dimensional solver family and do not resolve multidimensional mixing. The mesh audit does not establish convergence of the recorded response or reference states. The CH database and test designs have also informed model development, and seedwise uncertainty does not represent replication across independent physical databases. The Be extrapolation losses and CH high-tail losses remain part of the assessment. The method is validated here as a surrogate of specified numerical datasets, not as an experimentally validated predictor of ignition.

Computational reproducibility requires both scientific definitions and a controlled implementation. The accompanying package separates generic fit/predict operations from simulator-specific adapters, retains per-design predictions and manifests, and reproduces one complete accepted CH configuration at the original numerical tolerance. The thermal example permits data generation and training from scratch without proprietary material dependencies; regeneration of ICF trajectories requires the separately identified solver and tables.

Learning physical reference states improves design-only prediction in the principal CH comparisons, the thermal benchmark, and the random Be holdout. The additional value of GP experts is conditional on representation, sample size, and transfer setting. Together, these results establish an attributable physical-state-supervised architecture and a reproducible evaluation procedure, while identifying where a simpler global predictor and further physical-data verification are needed.

\section*{Acknowledgements}
Chang Liu is supported by the National Key R\&D Program of China (2022YFA1004500), the Beijing Natural Science Foundation (Z230003), the National Natural Science Foundation of China (12031001, 12102061), and the Presidential Foundation of the China Academy of Engineering Physics (YZJJZQ2022017).

The authors used AI-assisted tools to polish the language of the manuscript and to assist with routine scripting for data processing and figures. The core ideas, model formulation, software, numerical calculations, validation, and interpretation of the results are the work of the authors, who reviewed all assisted content and take full responsibility for this article.

\section*{Data and code availability}
The supporting design, response and state tables, train/test identifiers, per-design predictions, and evaluation records are openly available in Zenodo (version 1.0.0) at \url{https://doi.org/10.5281/zenodo.23253596}, under the Creative Commons Attribution 4.0 International licence (CC BY 4.0). The data archive includes file-level checksums and format descriptions. The \texttt{pssgp} implementation, examples, tests, and experiment utilities are openly available separately (version 0.1.0) at \url{https://doi.org/10.5281/zenodo.23255040}, under the MIT License. The generic package and independent thermal example require no ICF simulator. Archived ICF workflow utilities require the external dependencies documented in the software README; raw ICF trajectories and solver/material dependencies are excluded from both releases.

\section*{CRediT authorship contribution statement}
\textbf{Chang Liu:} Conceptualization, Methodology, Software, Investigation, Supervision, Writing---original draft.
\textbf{Hao Shen:} Methodology, Investigation, Writing---original draft, Writing---review \& editing.
\textbf{Chuanying Li:} Formal analysis, Data curation, Writing---review \& editing.

\section*{Declaration of competing interest}
The authors declare that they have no known competing financial interests or personal relationships that could have appeared to influence the work reported in this paper.

\end{document}